\documentclass{article} 
\usepackage{spconf,amsmath,graphicx}
\usepackage{url}
\usepackage{mathtools} 
\usepackage{amsmath}
\usepackage{graphicx}
\usepackage{amsfonts}
\usepackage{mathrsfs}
\usepackage{amsbsy}
\usepackage{amsthm}
\usepackage{bm}
\usepackage{subcaption}
\newtheorem{lemma}{Lemma}
\newtheorem{theorem}{Theorem}

\usepackage{tikz}
\usetikzlibrary{arrows,shapes,chains}
\usepackage[numbers,compress]{natbib} 

\title{IMPROVING SPIKING SPARSE RECOVERY VIA NON-CONVEX PENALTIES}
\name{Xiang Zhang$^{1}$, Lei Yu$^{1}$, Gang Zheng$^{2}$
}
\address{$^{1}$School of Electronic and Information, Wuhan University, China 
\\
$^{2}$Inria Lille, 40 Avenue Halley, 59650 Villeneuve d'Ascq, France}

\begin{document}
%
\maketitle
\begin{abstract}
Compared with digital methods, sparse recovery based on spiking neural networks has great advantages like high computational efficiency and low power-consumption. However, current spiking algorithms cannot guarantee more accurate estimates since they are usually designed to solve the classical optimization with convex penalties, especially the $\ell_{1}$-norm. In fact, convex penalties are observed to  underestimate the true solution in practice, while non-convex ones can avoid the underestimation.
Inspired by this, we propose an adaptive version of spiking sparse recovery algorithm to solve the non-convex regularized optimization, and provide an analysis on its global asymptotic convergence. Through experiments, the accuracy is greatly improved under different adaptive ways.
%
\end{abstract}
\section{Introduction}
As an inherent property of signals, sparsity has been widely exploited in many areas \cite{mairalSparseModelingImage2014}, and sparse recovery (SR) is one of the key problems.
Given an $M$-dimensional stimulus $\mathbf{s} \in \mathbb{R}^M$ (e.g., an $M$-pixel image), SR algorithms aim at recovering a $N$-dimensional sparse signal $\mathbf{a} \in \mathbb{R}^N$ from $\mathbf{s}$, and one of the most popular methods is to solve the regularized optimization problem
\begin{equation}\label{opt-general}
\underset{\mathbf{a}\in \mathbb{R}^N }{\operatorname{min}}\ \frac{1}{2}\|\mathbf{s}-\bm{\Phi} \mathbf{a}\|_{2}^{2}+\lambda C(\mathbf{a}),\quad \lambda>0,
\end{equation}
where $\bm{\Phi}=[\bm{\phi}_{1},\dots,\bm{\phi}_{N}] \in \mathbb{R}^{M \times N}$ with $M\ll N$ denotes the dictionary, and $C(\mathbf{a}):\mathbb{R}^N\rightarrow \mathbb{R}_{+}$ is a sparsity-inducing penalty. As the most sparsity-inducing function among convex ones \cite{brucksteinSparseSolutionsSystems2009}, $\ell_1$-norm is usually chosen as the penalty $C(\mathbf{a})=\|\mathbf{a}\|_1$, 
and many relevant iterative SR algorithms (ISR) have been proposed, including ISTA \cite{daubechiesIterativeThresholdingAlgorithm2004}, BPDN \cite{chenAtomicDecompositionBasisa}, LASSO \cite{tibshiraniRegressionShrinkageSelection1996a}, etc. However, these aforementioned algorithms usually need many iterations to converge and require large computing resources, which limits their applications in practice. To tackle this, the spiking neural networks (SNNs) have been widely exploited because of their computational efficiency \cite{zylberbergSparseCodingModel2011a,tangSparseCodingSpiking2017,tangConvergenceLCAFlows2016,chouAlgorithmicPowerSpiking2018}.

%


\par 
Inspired by the communication scheme that biological neurons use for efficient information transformation, 
spiking neurons (e.g., IAF \cite{gerstnerSpikingNeuronModels2002}) operate autonomously and only communicate rarely with other neurons via asynchronous spikes \cite{ghosh-dastidarSPIKINGNEURALNETWORKS2009}. These characteristics not only reduce the power consumption, which is desirable for practical usages, but also render the hardware implementation of SNN a potentially powerful computer.
Theoretically, the computability of SNN has been proved in pioneering works \cite{maass1996lower,maassNetworksSpikingNeurons1997,maass1998pulsed}, and several works have revealed the efficiency of SNN in solving specific problems such as the dictionary learning \cite{linSPARSEDICTIONARYLEARNING2019} and SR problems \cite{zylberbergSparseCodingModel2011a,tangSparseCodingSpiking2017,tangConvergenceLCAFlows2016,chouAlgorithmicPowerSpiking2018}.
Specifically, Shapero et al. are the first to implement the spiking SR (SSR) algorithm on an FPAA chip \cite{shaperoConfigurableHardwareIntegrate2013}. Though the SSR algorithm is restricted to non-negative variables to mimic biophysical variables (e.g., firing rates),
their experiments still verified the extremely high efficiency and low power-consumption of SNN in solving SR problems. 
\par 
However, though SSR algorithms are far superior to iterative ones in efficiency, current SSR algorithms cannot guarantee more accurate estimates. This is because most SSR systems are
designed to solve the classical $\ell_{1}$-norm regularized optimization. In fact, $\ell_{1}$-norm function tends to underestimate these high-amplitude components, and often results in biased estimates, while non-convex regularized optimization can lead to more accurate estimates since non-convex penalties will not punish these components excessively \cite{candesEnhancingSparsityReweighted2008,gassoRecoveringSparseSignals2009,hyderImprovedSmoothedEll2010}. 
This characteristic hints that the use of non-convex penalties might be a breakthrough to improve the accuracy of the SSR algorithm, which is also our motivation.

\par 
In this work, we will investigate the improvement of SSR from the perspective of optimization problems, and establish a bridge between the SNN and the non-convex regularized optimization. Through an adaptive mechanism, the proposed system encourages the recovery of high-amplitude components, thereby obtaining more accurate estimates. Moreover, we prove that our adaptive SSR algorithm (A-SSR) is globally asymptotically convergent, and will converge to the critical point of the non-convex regularized optimization. Based on above discussion, the relationship between SR algorithms 
and optimization problems is illustrated in Fig. \ref{draw} for clarity.


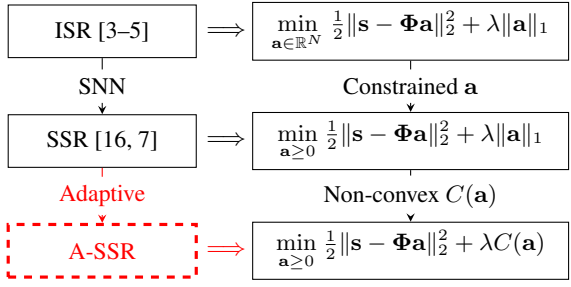
\begin{figure}
	\centering
	\tikzstyle{process} = [rectangle, minimum width=2.5cm, minimum height=0.7cm, text centered, draw=black]
	\tikzstyle{process2} = [rectangle, minimum width=4.2cm, minimum height=0.7cm, text centered, draw=black]
	\tikzstyle{noborder} = [rectangle, minimum width=2cm, minimum height=0.2cm, text centered, draw=none,fill=none]
	\tikzstyle{arrow} = [->,>=stealth]
	
	\begin{tikzpicture}
	[node distance=0.65cm]
	\small
	\node (ISR) [process] { ISR \cite{daubechiesIterativeThresholdingAlgorithm2004,chenAtomicDecompositionBasisa,tibshiraniRegressionShrinkageSelection1996a}};
	\node (SNN) [noborder,below of=ISR, yshift=-0.08cm] {SNN};
	\node (SSR) [process,below of=SNN, yshift=-0.08cm] {SSR \cite{shaperoConfigurableHardwareIntegrate2013,tangSparseCodingSpiking2017}};
	\node (Adaptive) [noborder,below of=SSR, yshift=-0.08cm,red] {Adaptive};
	\node (A-SSR) [process, below of=Adaptive, yshift=-0.08cm,dashed,very thick,red] {A-SSR};
	
	\node (conn1) [noborder,right of=ISR,xshift=1cm] {$\Longrightarrow$};
	\node (conn2) [noborder,right of=SSR,xshift=1cm] {$\Longrightarrow $};
	\node (conn3) [noborder,right of=A-SSR,xshift=1cm,red,dashed,very thick] {$\Longrightarrow$};
	
	\node (LASSO) [process2,right of=conn1,xshift=1.8cm] {$\underset{\mathbf{a}\in \mathbb{R}^N }{\operatorname{min}}\ \frac{1}{2}\|\mathbf{s}-\bm{\Phi} \mathbf{a}\|_{2}^{2}+\lambda \|\mathbf{a}\|_1$};
	\node (LASSO-1) [noborder,below of=LASSO, yshift=-0.08cm] {Constrained $\mathbf{a}$};
	\node (CLASSO) [process2,right of=conn2,xshift=1.8cm] {$\underset{\mathbf{a}\geq 0}{\operatorname{min}}\ \frac{1}{2}\|\mathbf{s}-\bm{\Phi} \mathbf{a}\|_{2}^{2}+\lambda \|\mathbf{a}\|_1$};
	\node (LASSO-2) [noborder,below of=CLASSO, yshift=-0.08cm] {Non-convex $C(\mathbf{a})$};
	\node (Non-convex) [process2,right of=conn3,xshift=1.8cm] {$\underset{\mathbf{a}\geq 0}{\operatorname{min}}\ \frac{1}{2}\|\mathbf{s}-\bm{\Phi} \mathbf{a}\|_{2}^{2}+\lambda C(\mathbf{a})$};
	
	\draw (LASSO) -- (LASSO-1);
	\draw [arrow] (LASSO-1) -- (CLASSO);
	\draw (CLASSO) -- (LASSO-2);
	\draw [arrow] (LASSO-2) -- (Non-convex);
	
	\draw (ISR) -- (SNN);
	\draw [arrow] (SNN) -- (SSR);
	\draw [red] (SSR) -- (Adaptive);
	\draw [arrow,red] (Adaptive) -- (A-SSR);
	
	\end{tikzpicture}
	
	\caption{.\ \  The relationship between SR algorithms and optimization problems. $A \Longrightarrow B$ represents that algorithm $A$ solves optimization $B$. Red part indicates our contribution.}
	\label{draw}
\end{figure}

\section{Adaptive Spiking Sparse Recovery}\label{chapter-2}
Under the assumption that each dictionary vector satisfies the unit norm $\|\bm{\phi}_{i}\|_{2}=1$, SSR is proved to converge to the solution of CLASSO problem \cite{tangSparseCodingSpiking2017}: ${\operatorname{min}_{\mathbf{a} \geq 0}} \frac{1}{2}\|\mathbf{s}-\bm{\Phi} \mathbf{a}\|_{2}^{2}+\lambda\|\mathbf{a}\|_{1}$. Following this, our target optimization can be formulated as  
\begin{equation}\label{concave-CLASSO}
\underset{\mathbf{a}\geq 0}{\operatorname{min}}\ E(\mathbf{a})=\frac{1}{2}\|\mathbf{s}-\bm{\Phi} \mathbf{a}\|_{2}^{2}+\lambda\sum_{i=1}^{N} g(|a_{i}|),
\end{equation}
where $g(\cdot)$ is a non-convex penalty. Next, we will first introduce the dynamic of the A-SSR algorithm, and then discuss about how to choose a suitable penalty $g(\cdot)$. Finally, we will prove the convergence of the A-SSR system, and point out that A-SSR will converge to the critical point of (\ref{concave-CLASSO}).
\subsection{The Dynamic of A-SSR Algorithm}\label{sec-1}
Based on the IAF model, each neuron in A-SSR will go through three states: \emph{inactive, active} and \emph{refractory period}.
Taking neuron-$i$ as an example, its potential $\nu_i(t)$ will be charged up according to the soma current $\mu_i(t)$ under inactive states, and neuron-$i$ will not become active until its potential exceeds the spiking threshold $\nu^{s}$, i.e. $\nu_i(t) > \nu^{s}$. When neuron-$i$ is active, its potential $\nu_{i}(t)$ will be immediately reset to the resting potential $\nu^{r}=0$, and simultaneously a spike signal will be sent to affect other neurons.
Afterward, neuron-$i$ will enter the refractory period $t^{ref}>0$, during which $\nu_{i}(t)$ remains at $\nu^{r}$. 
In the end, neuron-$i$ will return to the inactive state and start charging again. 
\par 
Similar to \cite{tangSparseCodingSpiking2017}, we define the soma current of neuron-$i$ as
\begin{equation}\label{mu}
\begin{aligned}
\mu_{i}(t) &=b_{i}-\sum_{j \neq i} \Omega_{i j}\left(\alpha * \sigma_{j}\right)(t),
\\
\tau\dot{\mu}_{i}(t)&=b_{i}-\mu_{i}(t)-\sum_{j \neq i} \Omega_{i j} \sigma_{j}(t),
\end{aligned}
\end{equation}
where $b_{i}\overset{\text { def }}{=}\bm{\phi}_{i}^{T} \mathbf{s}$ denotes the bias current; $\Omega_{i j}\overset{\text { def }}{=}\bm{\phi}_{i}^{T} \bm{\phi}_{j}$ measures the influence of spikes according to the similarity of receptive fields between neuron-$i$ and neuron-$j$; $\alpha(t)\stackrel{\text { def }}{=}e^{-t/\tau}{\rm H}(t)/\tau$ is a decay function, where $\tau$ is the time constant and ${\rm H}(t)$ denotes the Heaviside step function; $\sigma_{j}(t)\stackrel{\text { def }}{=}\sum_{k}\delta(t-t_{j,k})$ represents the spike train of neuron-$j$,  and \{$t_{j,k}$\} is its ordered spiking time sequence.
Equation (\ref{mu}) indicates that when no spike is generated, $\mu_{i}(t)$ is equal to the bias current $b_i$, and if other neurons become active, e.g., neuron-$j$, the spike will be weighted by $\Omega_{i j}$ first, and then affect the soma current of neuron-$i$ with an exponential decay.

\par 
Then, based on $\mu_{i}(t)$ and a non-convex penalty $g(\cdot)$, we can define the potential with consideration of the resetting mechanism:
\begin{equation}\label{potential}
\nu_{i}(t)\stackrel{\text { def }}{=}\int_{0}^{t}(\mu_{i}(s)-\Lambda_i(t))ds-\nu^{s}\int_{0}^{t} \sigma_i(s)ds,
\end{equation}
where $\Lambda_i(t)=\lambda g'(|a_{i}(t)|)$ indicates the leakage current and $g'(|a_i(t)|)$ is the first derivative of $g(|a_i(t)|)$ w.r.t. $|a_i(t)|$; $\lambda$ is the same constant in (\ref{concave-CLASSO}), and $a_{i}(t)$ is the firing rate of neuron-$i$, which is defined as
$
a_{i}(t)\stackrel{\text { def }}{=}\frac{1}{t} \int_{0}^{t}\sigma_{i}(s)ds
$.
Apart from $\mu_i(t)$, equation (\ref{potential}) hints that the speed of potential charging is also affected by the firing rate, since $\Lambda_i(t)$ will update adaptively according to $a_i(t)$. Moreover, equation (\ref{potential}) also establishes the connection between the chosen penalty $g(\cdot)$ and the adaptive way of A-SSR.
\par 
Based on above, our A-SSR algorithm can be fully described by  (\ref{mu}), (\ref{potential}) together with the definition of the firing rate $a_i(t)$ and the spike train $\sigma_i(t)$.
For simplicity, we set the spiking threshold $\nu^{s}=1$ from now on, and additionally set the lower bound of potential $\nu^{-}=0$, so that the potential $\nu_i(t)$ will be reset to $\nu^{-}$ once $\nu_i(t)<\nu^{-}$. 
In the next section, we will discuss the rules of choosing a suitable penalty $g(\cdot)$.

\subsection{Rules of Choosing Penalty $g(\cdot)$}\label{sec-2}
Based on previous works \cite{balavoineConvergenceNeuralNetwork2013a,chenConvergenceGuaranteesNonConvex2014,fossonBiconvexAnalysisLasso2018}, we summarize the rules as
\begin{enumerate}
	\item $g(\cdot)$ is non-negative and subanalytic on $[0, +\infty)$.
	\item The first derivative of $g(\cdot)$ is continuous and non-negative on $[0, +\infty)$, i.e. $g'(\cdot) \geq 0$.
	\item The second derivative of $g(\cdot)$ exists, and is bounded $g''(x)\in (-\frac{1}{\lambda},0), \forall x\in (0, +\infty)$.
	
\end{enumerate}
\par 
With rule-1, the objective $E(\mathbf{a})$ is also subanalytic, thus the convergence of A-SSR system can be guaranteed (see Section \ref{sec-3} for details). Rule-2  ensures that the term $\Lambda_i(t)$ will not go negative, in case $\nu_i(t)$ is overcharged. Last but not least, rule-3 states that the choice of penalties should be related to parameter $\lambda$, which is beneficial to the system stability.
Also, rule-3 implies that $g'(|a_i|)$ is monotonically decreasing over $|a_i|$ in our case. Hence, these neurons with higher firing rates are more likely to spike according to (\ref{potential}).
In other words, the penalty $g(\cdot)$ will not raise as fast as convex ones, thus these large elements can be better estimated. 
Following above, we can write the convergence result now.

\subsection{Convergence of A-SSR}\label{sec-3}
\begin{theorem}\label{main-theo}
	Under 
	a suitable penalty $g(\cdot)$ in Section \ref{sec-2}, we additionally define the average soma current $u_{i}(t)$ as
	\begin{equation}\label{u}
	u_{i}(t)\stackrel{\rm { def }}{=}\frac{1}{t} \int_{0}^{t} \mu_{i}(s)ds,
	\end{equation}
	then the average soma current $\mathbf{u}(t)$ and the firing rate $\mathbf{a}(t)$ of A-SSR are globally asymptotically convergent, and $\mathbf{a}(t)$ will converge to the critical point of (\ref{concave-CLASSO}) as time goes to infinity.
\end{theorem}
To prove it, we start with an auxiliary system designed as follows (the dependence on time is omitted for readability).
\begin{equation}\label{aux-system}
\begin{aligned}
\tau\dot{u}_{i}^{au}&=b_{i}-u_{i}^{au}-\sum_{j \neq i} \Omega_{i j} a_{j}^{au},
\\
a_i^{au}&=\max(u_i^{au}-\Lambda_i^{au},0)=\left\{\begin{array}{ll}
f\left(u_i^{au}\right), & \text{if}\ u_i^{au}>\Lambda^{au}_i, \\
0, & \text{else},
\end{array}\right. 
\end{aligned}
\end{equation}
where $f(u_i^{au}(t))=u_i^{au}(t)-\Lambda^{au}_i(t)=u_i^{au}(t)-\lambda g'(|a_i^{au}(t)|)$ is defined on $[0,+\infty)$, and the variables are restricted to non-negative (i.e. $u_i^{au}(t)\geq0$ and $a_i^{au}(t)\geq 0$).
Based on this, the proof can be divided into four lemmas. Lemma \ref{lemma-fir} and Lemma \ref{lemma-2} together prove the convergence property of the auxiliary system. Then, with the boundedness of $\bm{\mu}(t)$ and $\mathbf{u}(t)$ in Lemma \ref{lemma-1}, we can finally prove that A-SSR will converge to the auxiliary system as time goes to infinity, and the firing rate $\mathbf{a}(t)$ will converge to the critical point of (\ref{concave-CLASSO}).
\begin{lemma}\label{lemma-fir}
	With the objective $E(\mathbf{a}^{au})$ in (\ref{concave-CLASSO}), the auxiliary system satisfies that $\tau\dot{\mathbf{u}}^{au} \in -\partial E(\mathbf{a}^{au}).$
\end{lemma}
\begin{proof}
	By introducing the sign function $\operatorname{sgn}(\cdot)$, we have
	\begin{equation}\label{u-a}
	u_i^{au}(t)-a_i^{au}(t)\in \Lambda^{au}_i(t)\operatorname{sgn}(a_i^{au}(t))
	\end{equation}
	based on the following discussion:
	1) When $u_i^{au}(t)>\Lambda^{au}_i(t)$, $\operatorname{sgn}(a_i^{au}(t))=\operatorname{sgn}(u_i^{au}(t))=1$ holds since $\Lambda^{au}_i(t) \geq 0$ from the rule-2 in Section \ref{sec-2}. Hence, we can derive the relationship $u_i^{au}(t)-a_i^{au}(t)=\Lambda^{au}_i(t)=\Lambda^{au}_i(t)\operatorname{sgn}(a_i^{au}(t))$. 2) When $u_i^{au}(t)\in \left[0,\Lambda^{au}_i(t)\right]$, we have $a_i^{au}(t)=0$, therefore $u_i^{au}(t)-a_i^{au}(t)=u_i^{au}(t)\in \Lambda^{au}_i(t)\operatorname{sgn}(a_i^{au}(t))$. 
	\par 
	Then, according to (\ref{u-a}), the state $u_i^{au}(t)$ in (\ref{aux-system}) can be rewritten as $\tau\dot{\mathbf{u}}^{au}\in \mathbf{b}- \bm{\Omega} \mathbf{a}^{au} -\bm{\Lambda}^{au}\operatorname{sgn}(\mathbf{a}^{au})$, which is also equivalent to
	$
	\tau\dot{\mathbf{u}}^{au}\in \bm{\Phi}^{T}\mathbf{s}-\bm{\Phi}^{T} \bm{\Phi} \mathbf{a}^{au} -\lambda {g}'(|\mathbf{a}^{au}|) \operatorname{sgn}(\mathbf{a}^{au})
	$ since $\mathbf{b}=\bm{\Phi}^{T}\mathbf{s}$ and $\bm{\Omega}=\bm{\Phi}^{T} \bm{\Phi}$.
	On the other hand, the sub-differentiation of $E(\mathbf{a}^{au})$ w.r.t. $\mathbf{a}^{au}$ yields
	$
	\partial E(\mathbf{a}^{au})=\left( -\bm{\Phi}^{T}\mathbf{s}+\bm{\Phi}^{T} \bm{\Phi} \mathbf{a}^{au} +\lambda {g}'(|\mathbf{a}^{au}|) \operatorname{sgn}(\mathbf{a}^{au}) \right)^{T},
	$
	then it is straightforward that  
	$
	\tau\dot{\mathbf{u}}^{au} \in -\partial E(\mathbf{a}^{au})
	$ holds.
\end{proof}
Next, we are going to prove the convergence of the auxiliary system. In fact, the global asymptotic convergence of such systems has been proved in \cite{balavoineConvergenceNeuralNetwork2013a} using Łojasiewicz Inequality. Here, we focus on proving the conditions required by Theorem 1 and Theorem 2 in \cite{balavoineConvergenceNeuralNetwork2013a}.

\begin{lemma}\label{lemma-2}
	The state $\mathbf{u}^{au}(t)$ and the output $\mathbf{a}^{au}(t)$ of the auxiliary system (\ref{aux-system}) are globally asymptotically convergent.
\end{lemma}
\begin{proof}
	First of all, (\ref{u-a}) corresponds to the condition (3) in \cite{balavoineConvergenceNeuralNetwork2013a}, since $C(a_i^{au}(t))= g(|a_i^{au}(t)|)$ in our case.
	Additionally, we have the first derivative of $f(u_i^{au}(t))$ when $u_i^{au}(t)>\Lambda_i^{au}(t)$:	
	$$
	\begin{aligned}
	f'(u_i^{au}(t)) &= 1 - \frac{\mathrm{d}\Lambda_i(t)}{\mathrm{d}a_i^{au}(t)} \frac{\mathrm{d}a_i^{au}(t)}{\mathrm{d}u_i^{au}(t)}
	\\
	&=1-\lambda g''(|a_i^{au}(t)|)\operatorname{sgn}(a_i^{au}(t)) f'(u_i^{au}(t)).
	\end{aligned}
	$$
	With the rule-3 in Section \ref{sec-2}, one can obtain that 
	$
	f'(u_i^{au}(t))=1/(1+\lambda g''(|a_i^{au}(t)|)\operatorname{sgn}(a_i^{au}(t)))>0,
	$
	which corresponds to the condition (5b) in \cite{balavoineConvergenceNeuralNetwork2013a}. Also, $f'(\cdot)$ satisfies the condition (7) in \cite{balavoineConvergenceNeuralNetwork2013a}, since one can always find an upper bound of $f'(\cdot)$ on bounded intervals. As for the condition (6) in \cite{balavoineConvergenceNeuralNetwork2013a}, it is used to guarantee that $E(\mathbf{a})$ is subanalytic, which is also met by the rule-1 in Section \ref{sec-2}.
	For the remaining conditions, there might be some small differences due to the one-side case we are concerned about, but they will not affect the final result of global asymptotic convergence.
\end{proof}
\begin{figure*}[ht]
	\centering
	\begin{subfigure}{0.26\textwidth}
		\includegraphics[width=\textwidth,height=0.618\textwidth]{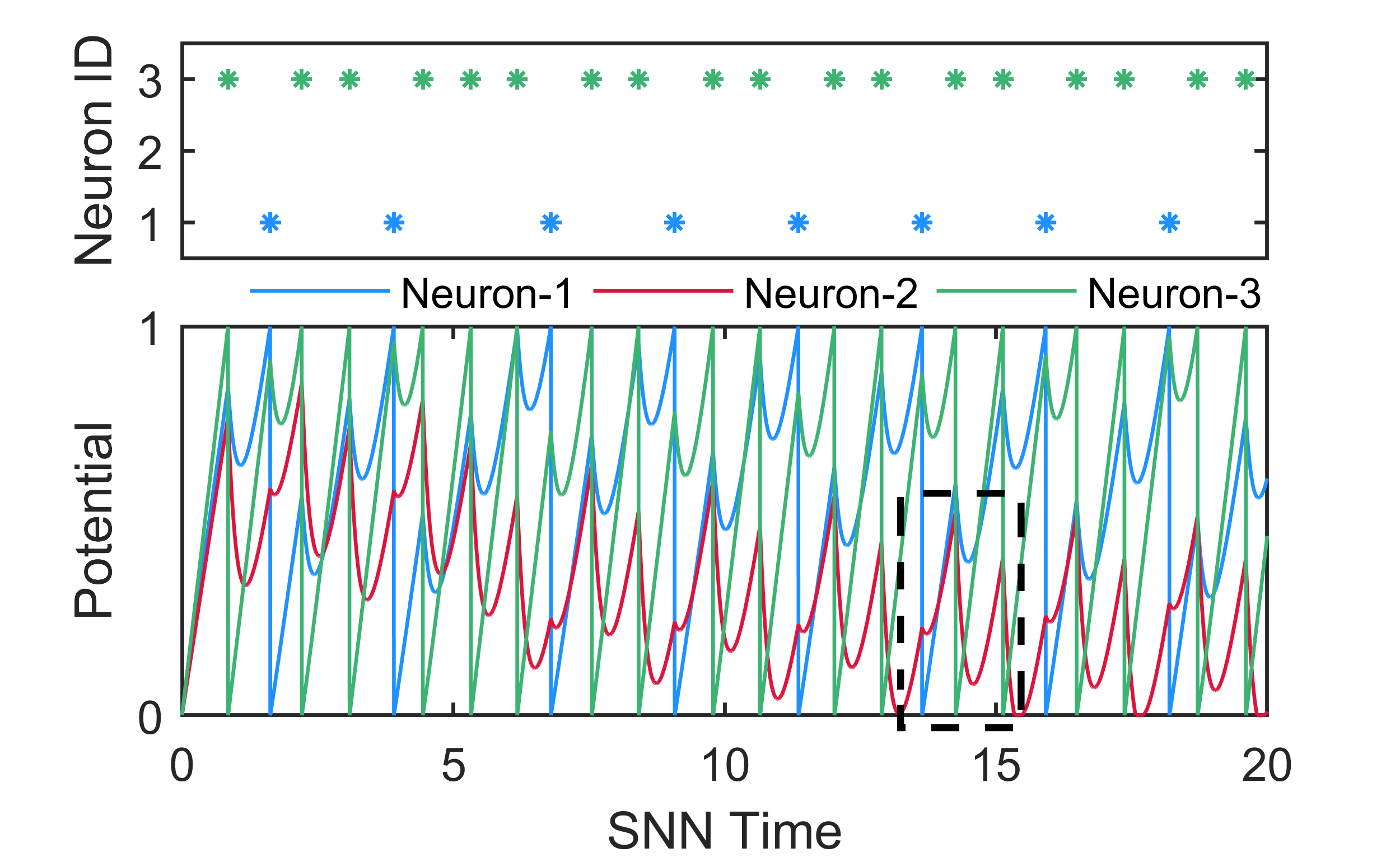}
		\caption{\rmfamily \fontsize{8pt}{0} SSR}
		\label{SLCA}
	\end{subfigure}
	\begin{subfigure}{0.26\textwidth}
		\includegraphics[width=\textwidth,height=0.618\textwidth]{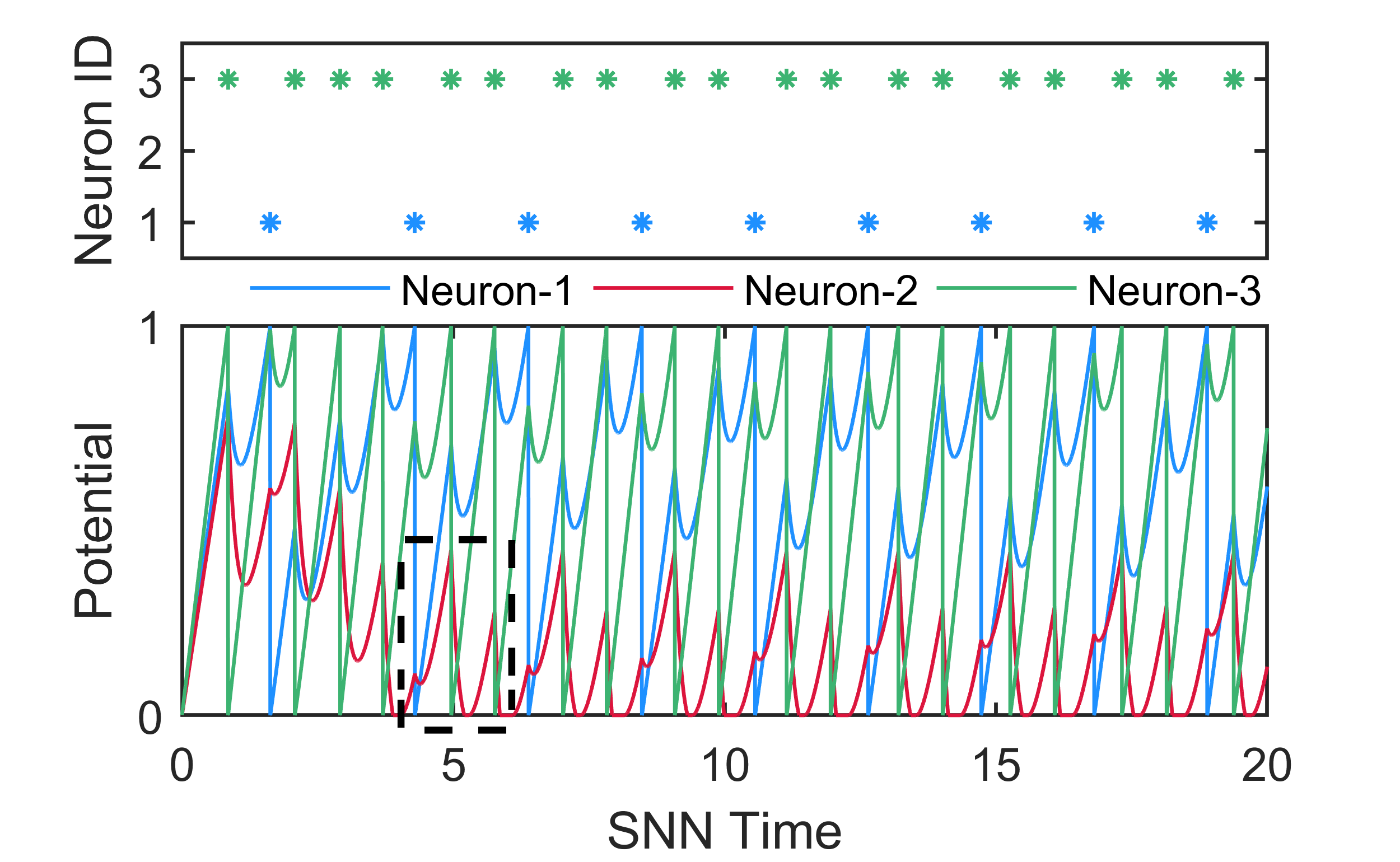}
		\caption{\rmfamily \fontsize{8pt}{0} A-SSR}
		\label{A-SSR}
	\end{subfigure}
	\begin{subfigure}{0.26\textwidth}
		\includegraphics[width=\textwidth,height=0.618\textwidth]{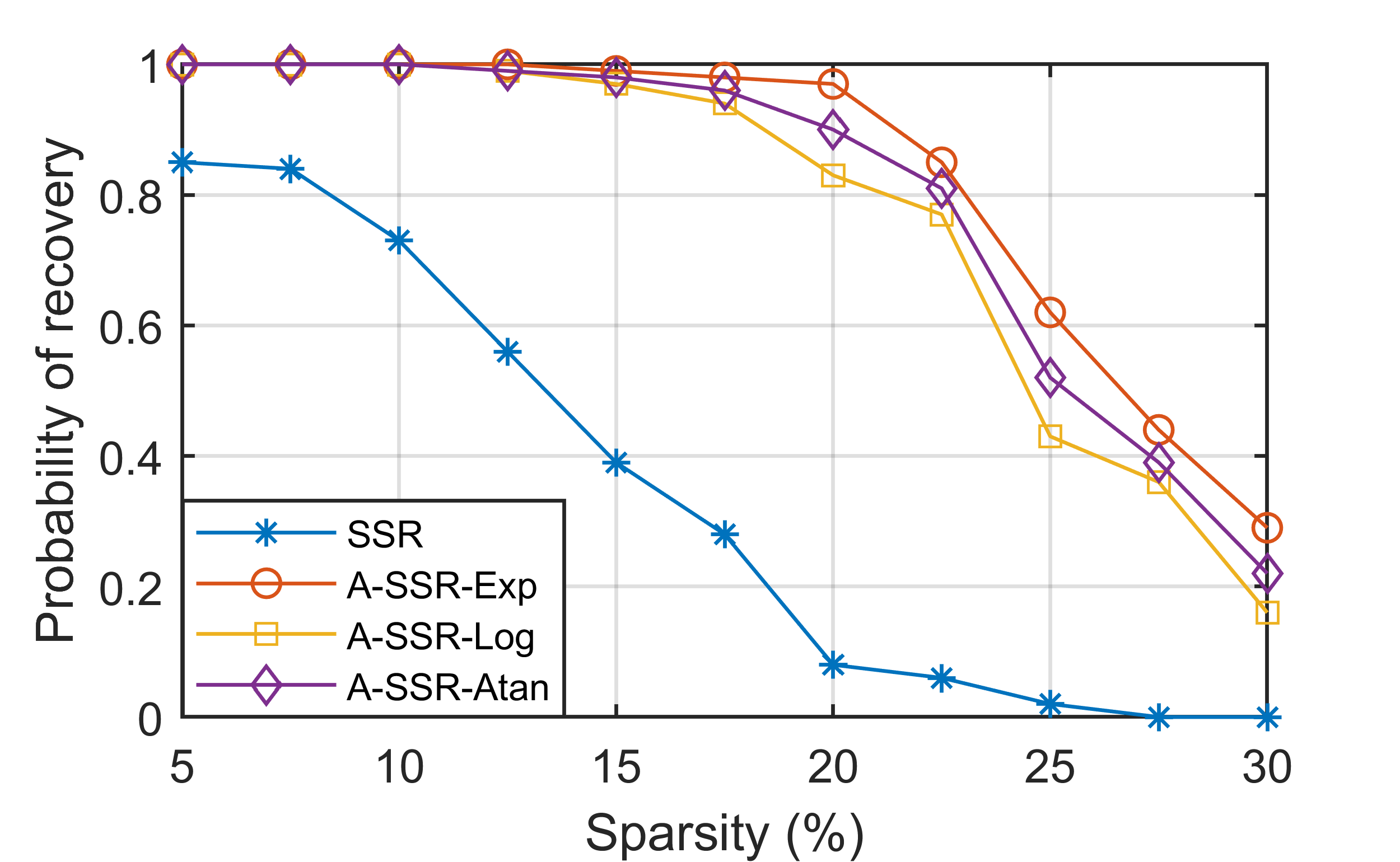}
		\caption{\rmfamily \fontsize{8pt}{0} Probability of Recovery}
		\label{Prob}
	\end{subfigure}
	\begin{subfigure}{0.22\textwidth}
		\includegraphics[width=\textwidth,height=0.75\textwidth]{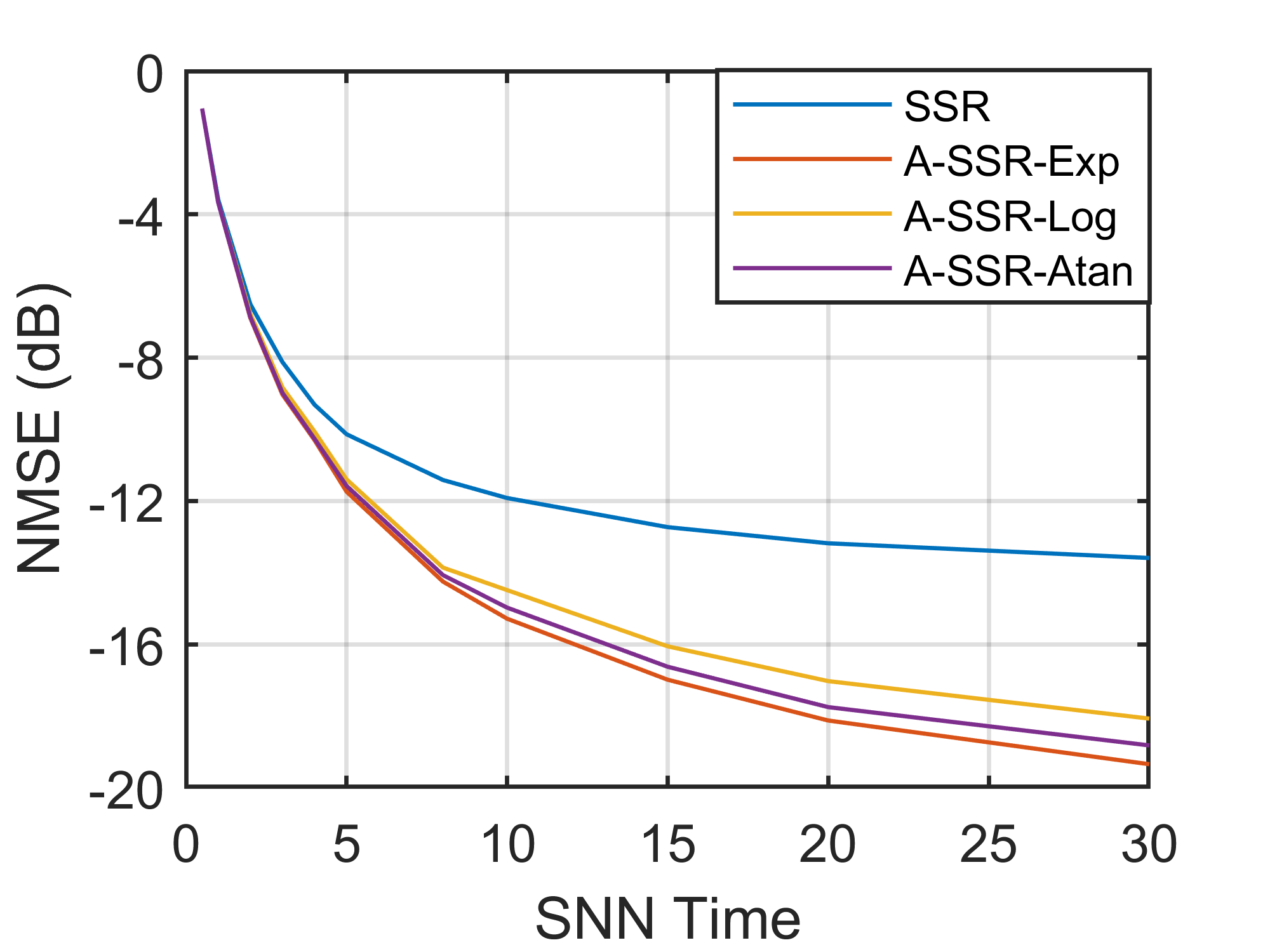}
		\caption{\rmfamily \fontsize{8pt}{0} Convergence}
		\label{speed}
	\end{subfigure}
	\begin{subfigure}{0.22\textwidth}
		\includegraphics[width=\textwidth,height=0.75\textwidth]{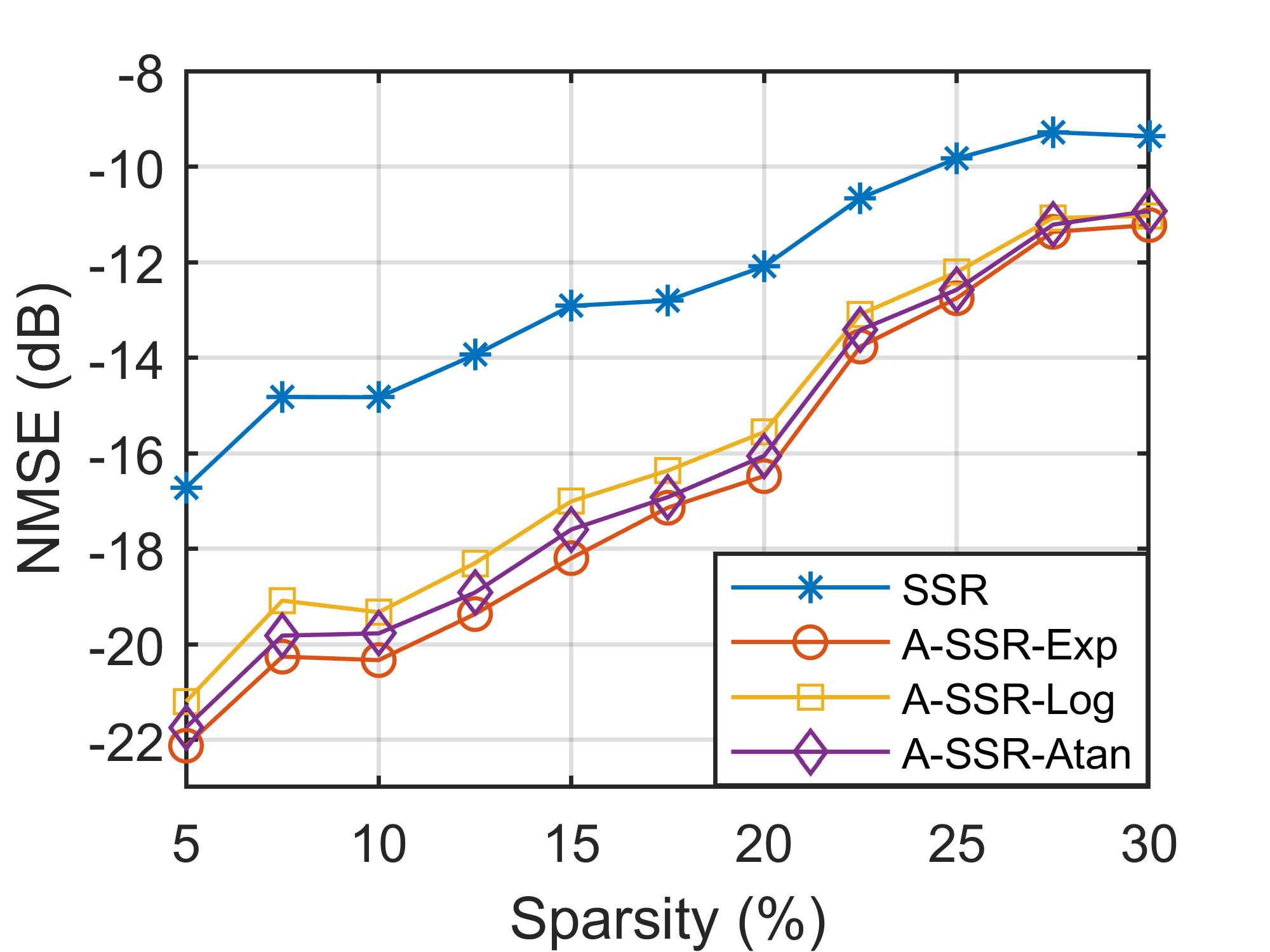}
		\caption{\rmfamily \fontsize{8pt}{0} Sparsity}
		\label{sparsity}
	\end{subfigure}
	\begin{subfigure}{0.22\textwidth}
		\includegraphics[width=\textwidth,height=0.75\textwidth]{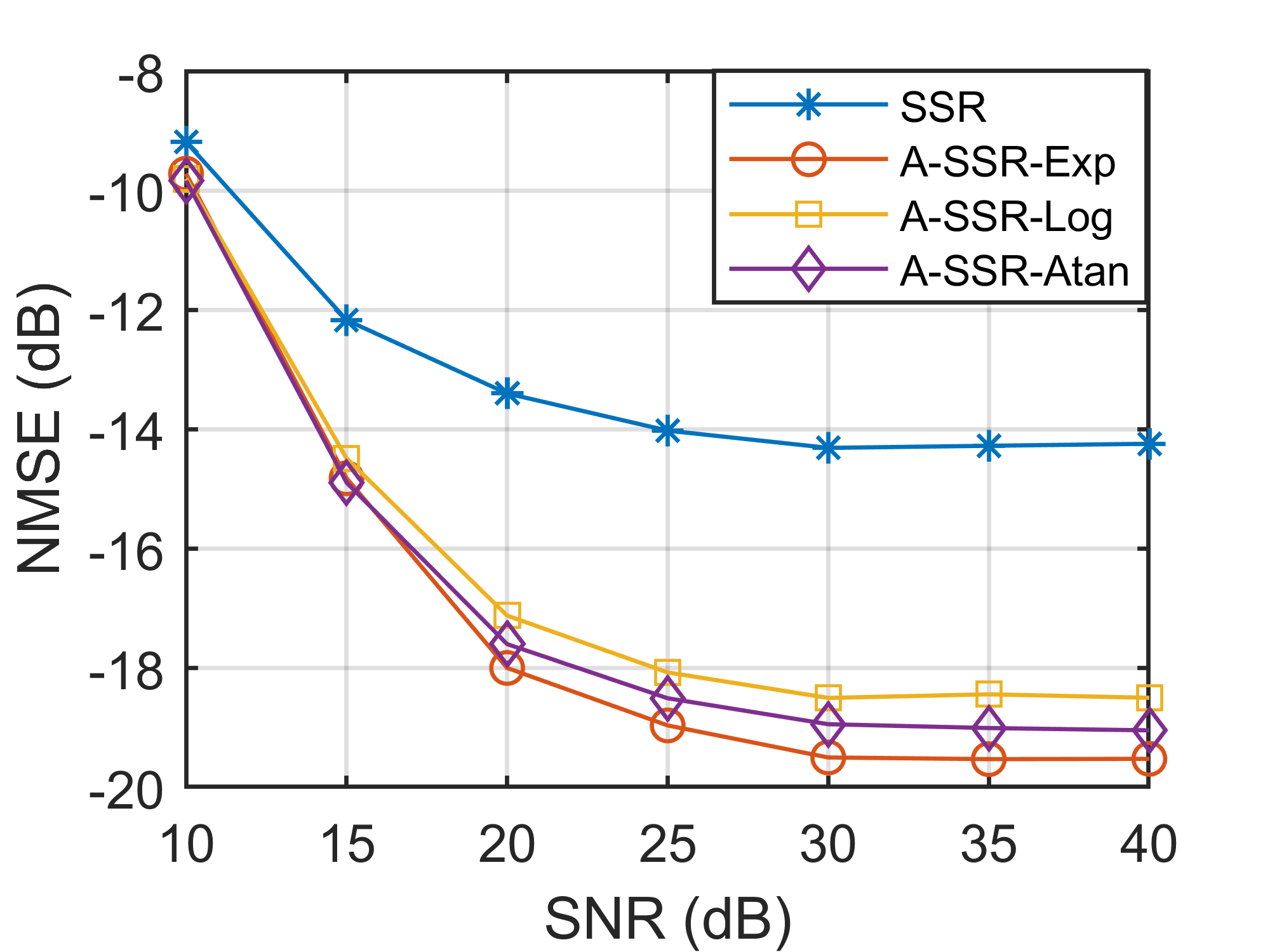}
		\caption{\rmfamily \fontsize{8pt}{0} Noise}
		\label{snr}
	\end{subfigure}
	\begin{subfigure}{0.22\textwidth}
		\includegraphics[width=\textwidth,height=0.75\textwidth]{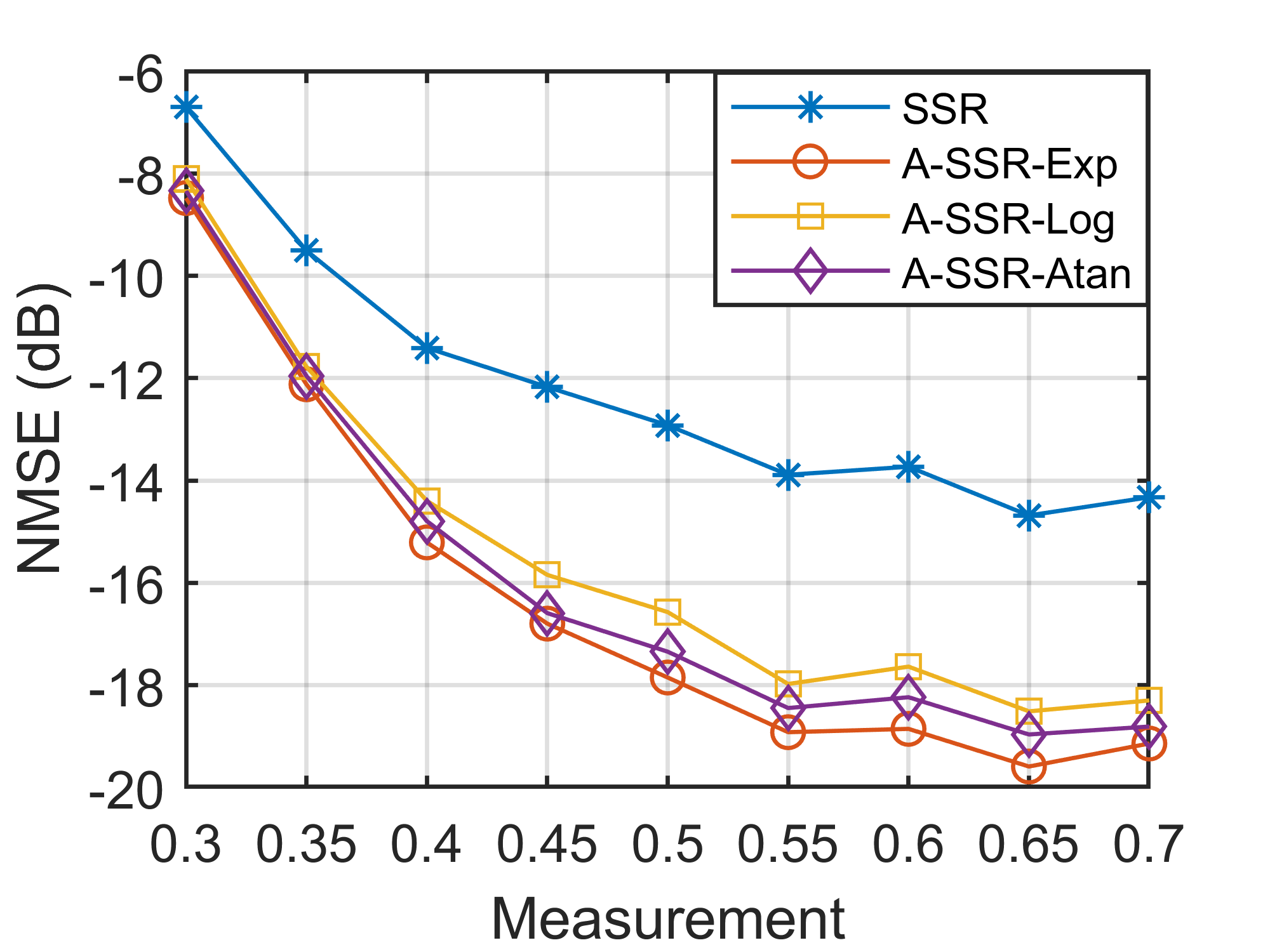}
		\caption{\rmfamily \fontsize{8pt}{0} Measurement}
		\label{measurement}
	\end{subfigure}
	\caption{.\ \ The performance comparisons. (a), (b) The potential plot (bottom) and spike raster plot (top) of SSR and the A-SSR with exponential penalty. SNN time is relative to the absolute time by $\tau$. Dashed box indicates the first stable period. (c) is the probability of successful recovery in the noise-free scenario.
    (d) is the convergence comparison. $NMSE=10 \log _{10} (\| \hat{\mathbf{a}}-\mathbf{a}\|_{2}^{2}/\|\hat{\mathbf{a}}
    \|_{2}^{2})$, where $\hat{\mathbf{a}}$ denotes the original signal. (e)-(g) are the performance under different sparsity, noise, and measurement. Sparsity is the percentage of non-zero coefficients in total. Measurement represents the ratio $M/N$ of dictionary $\bm{\Phi} \in \mathbb{R}^{M \times N}$.
    }
\end{figure*}
\begin{lemma}\label{lemma-1}
	There exist an upper bound $B_+$ and a lower bound $B_-$ such that $\mu_{i}(t), u_{i}(t) \in [B_{-},B_{+}], \forall i, t \geq 0$.
\end{lemma}
\begin{proof}
	Based on the assumption that dictionary atoms are  normalized to unit norm,
	we define the auxiliary parameters:
	$\Omega_{max}\stackrel{\text { def }}{=} \max_{j \neq i} \left|\bm{\phi}_{i}^{T} \bm{\phi}_{j}\right|$ and
	$b_{max} \stackrel{\text { def }}{=} \max_{i} \left| b_{i} \right|$.
	Then, with the refractory period in spiking neurons, the duration between spikes cannot be arbitrarily small whenever these two spike times exist, i.e. $t_{i,k+1}-t_{i,k} \geq t^{ref}, \forall k \geq 0$. Therefore, we have the inequality
	$
	\beta \stackrel{\text { def }}{=} \frac{1}{\tau}\sum_{c=0}^{\infty} e^{-(ct^{ref}/\tau)} \geq (\alpha * \sigma)(t),
	$
	and the following relationship can be obtained 
	based on (\ref{mu})
	$$
	\mu_{i}(t)\leq b_{max}+\left(N-1\right) \Omega_{max}\beta,
	$$
	where $N$ denotes the total number of neurons in A-SSR. 
	Similarly, the relationship $\mu_{i}(t) \geq -b_{max}-(N-1)\Omega_{max} \beta$ can be derived. Since $u_i$ is the average of $\mu_i$, there exist bounds $B_{+}=b_{max}+\left(N-1\right) \Omega_{max} \beta$ and $B_{-}=-b_{max}-(N-1)\Omega_{max} \beta$ such that $\mu_{i}(t), u_{i}(t) \in [B_{-},B_{+}],\forall i,t \geq 0$.
\end{proof}

\begin{lemma}\label{lemma-4}
	As $t\rightarrow\infty$, the average soma current $\mathbf{u}(t)$ and the firing rate $\mathbf{a}(t)$ of A-SSR is globally asymptotically convergent, and $\mathbf{a}(t)$ will converge to the critical point of (\ref{concave-CLASSO}).
\end{lemma}
\begin{proof}
	With the derivative derived from (\ref{u}): 
	\begin{equation}\label{u-dot}
	\dot{u}_{i}(t)=\frac{\mu_{i}(t)-u_{i}(t)}{t},
	\end{equation}  
	we have the equation below by applying  $t^{-1}\int_{0}^{t}	ds$ to $\dot{\mu}$ in (\ref{mu})
	\begin{equation}\label{dot-u-3}
	\tau\dot{u}_{i}(t)=b_{i}-u_{i}(t)-\sum_{j \neq i} \Omega_{i j} a_{j}(t)-\tau\frac{u_{i}(t)-\mu_{i}(0)}{t}.
	\end{equation}
	On the other hand, based on $\nu^{s}=1$ and the definition of $a_i(t)$, we can derive the following relationship from (\ref{potential}):
	\begin{equation}\label{u-a-temp-3}
	\begin{aligned}
	a_i(t)&=\frac{1}{t}\int_{0}^{t}(\mu_{i}(s)-\Lambda_i(t))ds - \frac{\nu_{i}(t)}{t}
	\\
	&=u_i(t)-\Lambda_i(t)-\frac{\nu_{i}(t)}{t}.
	\end{aligned}
	\end{equation}	
	Under $a_i(t)\geq 0$, equation (\ref{u-a-temp-3}) can be transferred to
	\begin{equation}\label{u-a-max-3}
	a_i(t)=\max(u_i(t)-\Lambda_i(t)-\nu_i(t)/t,\ 0).
	\end{equation}
	With $\nu_i(t)\in [\nu^-,\nu^{s}]$ and  Lemma \ref{lemma-1}, it is obvious that (\ref{dot-u-3}) and (\ref{u-a-max-3}) will converge to (\ref{aux-system}) as $t \rightarrow \infty$, i.e. the A-SSR will converge to the auxiliary system. Therefore, with Lemma \ref{lemma-2}, the $\mathbf{u}(t)$ and $\mathbf{a}(t)$ of A-SSR are also globally asymptotically 
	convergent when time goes to infinity.
	\par 
	Assume that $\mathbf{u}^*$ is the stable point of A-SSR and $\mathbf{a}^*$ is the corresponding output. With (\ref{u-dot})
	and Lemma \ref{lemma-1}, we have
	\begin{equation}\label{u=0}
	\dot{\mathbf{u}}^*=\lim _{t \rightarrow \infty} \dot{\mathbf{u}}(t)=\lim _{t \rightarrow \infty}\frac{\bm{\mu}(t)-\mathbf{u}(t)}{t}=0.
	\end{equation}
	Since A-SSR converges to the auxiliary system, the A-SSR system satisfies $\tau\dot{\mathbf{u}}(t) \in -\partial E(\mathbf{a}(t))$ as $t\rightarrow\infty$. Combining (\ref{u=0}), it is clear that $0\in \partial E(\mathbf{a}^{*})$, which means that as $t \rightarrow \infty$, $\mathbf{a}(t)$ will converge  to $\mathbf{a}^{*}$, the critical point of (\ref{concave-CLASSO}).
\end{proof}
With above results, we complete the proof of Theorem \ref {main-theo}.

\section{Numerical Experiment}\label{chapter-3}
Finally, to verify the performance of A-SSR, we present numerical experiments on Nengo, a neural engineering simulation platform \cite{bekolayNengoPythonTool2014}.
At first, 
we compare the working processes between SSR and A-SSR with exponential penalty $g(|a_i|)=1-e^{-\gamma|a_i|}$, where $\gamma=1$. For simplicity, we design a 3-neuron SNN, and set the original signal as $\hat{\mathbf{a}}=[0.4792,0,0.9754]^T$.
In Fig. \ref{SLCA}, \ref{A-SSR}, these active neurons (e.g., neuron-1) in A-SSR tend to produce more spikes compared to the SSR system, since the adaptive mechanism encourages these active neurons to spike more frequently. This characteristic not only leads to better accuracy, but also makes A-SSR be in periodical stabilization earlier, i.e. the system will converge faster.
In the end, the normalized mean square errors ($NMSEs$) generated by SSR and A-SSR are $-19.9461(dB)$ and $-28.9681(dB)$ correspondingly.
\par 
Next, to explore the performance of A-SSR with different adaptive ways, we add A-SSR with $g(|a_i|)=\log(|a_i|+\epsilon)$ and $g(|a_i|)=\arctan(|a_i|/\eta)$ into experiments, and set constants $\epsilon=1$, $\eta=1$. 
In Fig. \ref{Prob}, we define successful recovery the case $NMSE<-15(dB)$ under the dictionary $\bm{\Phi}^{100\times 200}$, and compare the probability of successful recovery 
by varying sparsity.
Afterward, we consider the situation where the signal is corrupted by white Gaussian noise. With these default settings: $\bm{\Phi} \in \mathbb{R}^{100\times 200}$, $SNR=20(dB)$ and sparsity $=15\%$, we present the experiments from four aspects: Fig. \ref{speed} shows the convergence comparison under the same conditions, and Fig. \ref{sparsity}, \ref{snr}, \ref{measurement} show the performance under different levels of sparsity, noise, and measurement respectively.

\par 
Evidently, all the A-SSR algorithms achieve a general improvement on accuracy compared to SSR. Besides, the performance of A-SSR will vary with the chosen non-convex penalty, i.e. the adaptive way, and the exponential penalty is the best choice for the proposed experiments.
\section{Conclusion}\label{chapter-4}

We investigated the improvement of SSR based on the optimization problem. By establishing the connection between the non-convex penalty and the adaptive mechanism, our A-SSR system can avoid the underestimation of high-amplitude components, thereby obtaining more accurate estimates.

\vfill\pagebreak

%
%

\bibliography{library2}
\bibliographystyle{IEEEbib}
\end{document}